\documentclass[11pt]{article}

\newcommand{\cmp}{Comm. Math. Phys.~}

\newcommand{\jmp}{J. Math. Phys.~}
\newcommand{\jpa}{J. Phys. A~}

\newcommand{\pra}{Phys. Rev. A~}

\usepackage{epstopdf}
\usepackage{color}
\usepackage{subfigure}
\usepackage[T1]{fontenc}
\usepackage[sc]{mathpazo}
\usepackage{amsmath}
\usepackage{amssymb}
\usepackage{graphicx,color}
\usepackage{framed}
\usepackage{multirow}
\usepackage{enumerate}
\usepackage{amsthm}
\usepackage{amsfonts,mathrsfs}
\usepackage{geometry} 
\usepackage{eepic}
\usepackage{ifthen}
\usepackage[vcentermath]{youngtab}
\usepackage[unicode=true,pdfusetitle, bookmarks=true,bookmarksnumbered=false,bookmarksopen=false, breaklinks=false,pdfborder={0 0 0},backref=false,colorlinks=false] {hyperref}
\hypersetup{
colorlinks,linkcolor=myurlcolor,citecolor=myurlcolor,urlcolor=myurlcolor}
\definecolor{myurlcolor}{rgb}{0,0,0.7}

\usepackage{color}

\newcommand{\red}{\textcolor{red}}

\newcommand{\proj}[1]{| #1\rangle\!\langle #1 |}

\usepackage{hyperref}
\hypersetup{pdfpagemode=UseNone}

\newcommand{\tinyspace}{\mspace{1mu}}

\newcommand{\op}[1]{\operatorname{#1}}

\newcommand{\abs}[1]{\left\lvert\tinyspace #1 \tinyspace\right\rvert}

\newcommand{\norm}[1]{\left\lVert\tinyspace #1 \tinyspace\right\rVert}

\renewcommand{\det}{\operatorname{det}}
\renewcommand{\t}{{\scriptscriptstyle\mathsf{T}}}

\newcommand{\im}{\op{im}}

\newcommand{\sign}{\op{sign}}

\def\vol{\mathrm{vol}}

\def \dif {\mathrm{d}}
\def \diag {\mathrm{diag}}
\def \vol {\mathrm{vol}}
\def \re {\mathrm{Re}}
\def \im {\mathrm{Im}}

\def\complex{\mathbb{C}}
\def\real{\mathbb{R}}
\def\natural{\mathbb{N}}

\def\I{\mathbb{1}}

\newenvironment{mylist}[1]{\begin{list}{}{
    \setlength{\leftmargin}{#1}
    \setlength{\rightmargin}{0mm}
    \setlength{\labelsep}{2mm}
    \setlength{\labelwidth}{8mm}
    \setlength{\itemsep}{0mm}}}
    {\end{list}}

\def\ot{\otimes}

\newcommand{\inner}[2]{\langle #1 , #2\rangle}
\newcommand{\iinner}[2]{\langle #1 | #2\rangle}

\newcommand{\Inner}[2]{\left\langle #1 , #2\right\rangle}
\newcommand{\Innerm}[3]{\left\langle #1 \left| #2 \right| #3 \right\rangle}

\newcommand{\pa}[1]{(#1)}
\newcommand{\Pa}[1]{\left(#1\right)}

\newcommand{\Br}[1]{\left[#1\right]}
\newcommand{\set}[1]{\{#1\}}
\newcommand{\Set}[1]{\left\{#1\right\}}

\newcommand{\ket}[1]{|#1\rangle}

\DeclareMathOperator{\vectorize}{vec}
\newcommand{\col}[1]{\vectorize\pa{#1}}

\DeclareMathOperator{\trace}{Tr}

\newcommand{\Ptr}[2]{\trace_{#1}\Pa{#2}}

\newcommand{\Tr}[1]{\Ptr{}{#1}}

\def\cF{\mathcal{F}}\def\cH{\mathcal{H}}

\def\bP{\mathbf{P}}

\def\bss{\boldsymbol{s}}\def\bst{\boldsymbol{t}}
\def\bsu{\boldsymbol{u}}\def\bsv{\boldsymbol{v}}\def\bsx{\boldsymbol{x}}\def\bsy{\boldsymbol{y}}
\def\bsz{\boldsymbol{z}}

\def\rP{\mathrm{P}}
\def\rU{\mathrm{U}}\def\rV{\mathrm{V}}

\newtheorem{thrm}{Theorem}[section]

\newtheorem{prop}[thrm]{Proposition}
\newtheorem{cor}[thrm]{Corollary}
\theoremstyle{definition}
\newtheorem{definition}[thrm]{Definition}
\newtheorem{remark}[thrm]{Remark}
\newtheorem{exam}[thrm]{Example}

\numberwithin{equation}{section}

\newcounter{questionnumber}

\usepackage[all]{xy}

\usepackage{color}

\begin{document}

\title{Dirac Delta Function of Matrix Argument}

\author{Lin Zhang\footnote{E-mail: godyalin@163.com; linyz@zju.edu.cn}\\
  {\it\small Institute of Mathematics, Hangzhou Dianzi University, Hangzhou 310018, PR~China}}

\date{}
\maketitle \mbox{}\hrule\mbox{}
\begin{abstract}

Dirac delta function of matrix argument is employed frequently in
the development of diverse fields such as Random Matrix Theory,
Quantum Information Theory, etc. The purpose of the article is
pedagogical, it begins by recalling detailed knowledge about
Heaviside unit step function and Dirac delta function. Then its
extensions of Dirac delta function to vector spaces and matrix
spaces are discussed systematically, respectively. The detailed and
elementary proofs of these results are provided. Though we have not
seen these results formulated in the literature, there certainly are
predecessors. Applications are also mentioned.

\end{abstract}
\mbox{}\hrule\mbox{}

\section{Heaviside unit step function $H$ and Dirac delta function $\delta$}

The materials in this section are essential from Hoskins' Book
\cite{Hoskins2009}. There are also no new results in this section.
In order to be in a systematic way, it is reproduced here.

The Heaviside unit step function $H$ is defined as
\begin{eqnarray}
H(x) := \begin{cases} 1,& x>0,\\0,& x<0.\end{cases}
\end{eqnarray}
That is, this function is equal to 1 over $(0,+\infty)$ and equal to
0 over $(-\infty,0)$. This function can equally well have been
defined in terms of a specific expression, for instance
\begin{eqnarray}
H(x) = \frac12\Pa{1+\frac{x}{\abs{x}}}.
\end{eqnarray}
The value $H(0)$ is left undefined here. For all $x\neq0$,
$$
H'(x)=\frac{\dif H(x)}{\dif x}=0
$$
corresponding to the obvious fact that the graph of the function
$y=H(x)$ has zero slope for all $x\neq0$. Naturally we describe the
slope as "infinite" at origin. We shall denote by $\delta(x)$ the
derivative of $H(x)$:
$$
\delta(x)=H'(x)=0,\quad\forall x\neq0,
$$
and
$$
\delta(0) = +\infty.
$$
We recall the definition of Dirac delta function:
\begin{definition}[Dirac delta function]
Dirac delta function $\delta(x)$  is defined by
\begin{eqnarray}
\delta(x) = \begin{cases} +\infty,&\text{if } x=0;\\0, &\text{if }
x\neq 0.
\end{cases}
\end{eqnarray}
\end{definition}

\begin{prop}[Sampling property of the Dirac delta function]
If $f$ is any function which is continuous on a neighborhood of $0$,
then
\begin{eqnarray}
\int^{+\infty}_{-\infty} f(x)\delta(x)\dif x = f(0).
\end{eqnarray}
\end{prop}
In fact, we have, for $a\neq0$,
$$
\int^{+a}_{-a}f(x)\delta(x)\dif x = f(0)
$$
and
$$
\int^a_{-\infty}\delta(x)\dif x=
\int^{+\infty}_{-\infty}H(a-x)\delta(x)\dif x = H(a).
$$
\begin{definition}
Assume that $f$ is continuous function which vanishes outside some
finite interval. There corresponds a certain number which we write
as $\Inner{H}{f}$, given by
\begin{eqnarray}
\Inner{H}{f}: = \int^{+\infty}_{-\infty} f(x)H(x)\dif x =
\int^{+\infty}_0 f(x)\dif x.
\end{eqnarray}
Similarly, $\Inner{\delta}{f}$ is given by
\begin{eqnarray}
\Inner{\delta}{f}: = \int^{+\infty}_{-\infty} f(x)\delta(x)\dif x =
f(0).
\end{eqnarray}
\end{definition}
For an ordinary function $f$ and a fixed $a\in\real$, the symbol
$f_a$ denotes the translation of $f$ with respect to $a$:
\begin{eqnarray}
f_a(x):=f(x-a).
\end{eqnarray}
Thus
\begin{eqnarray}
\Inner{\delta_a}{f} = \int^{+\infty}_{-\infty} f(x)\delta_a(x)\dif x
= \int^{+\infty}_{-\infty} f(x)\delta(x-a)\dif x =f(a).
\end{eqnarray}
From the above discussion, it is easy to see that
\begin{eqnarray}
f(x)\delta(x-a) = f(x)\delta_a(x) = f(a)\delta(x).
\end{eqnarray}
This fact will be used later.
\begin{prop}
If $f$ is any function which has a continuous derivative $f'$, at
least in some neighborhood of $0$, then
\begin{eqnarray}
\Inner{\delta'}{f}: = \int^{+\infty}_{-\infty} f(x)\delta'(x)\dif x
= -\Inner{\delta}{f'} = -f'(0).
\end{eqnarray}
\end{prop}

\begin{proof}
Since
\begin{eqnarray*}
\int^{+\infty}_{-\infty}
f(x)\frac{\delta(x)-\delta(x-\epsilon)}{\epsilon}\dif x &=&
\frac1{\epsilon}\Pa{\int^{+\infty}_{-\infty} f(x)\delta(x)\dif x -
\int^{+\infty}_{-\infty} f(x)\delta(x-\epsilon)\dif x}\\
&=& \frac{f(0)-f(\epsilon)}{\epsilon} =
-\frac{f(\epsilon)-f(0)}{\epsilon},
\end{eqnarray*}
it follows that
\begin{eqnarray}
\Inner{\delta'}{f} &=& \int^{+\infty}_{-\infty} f(x)\delta'(x)\dif
x=\int^{+\infty}_{-\infty}
f(x)\lim_{\epsilon\to0}\frac{\delta(x)-\delta(x-\epsilon)}{\epsilon}\dif
x\\
&=&\lim_{\epsilon\to0}\int^{+\infty}_{-\infty}
f(x)\frac{\delta(x)-\delta(x-\epsilon)}{\epsilon}\dif x\\
&=&\lim_{\epsilon\to0}-\frac{f(\epsilon)-f(0)}{\epsilon} =
-f'(0)=-\Inner{\delta}{f'}.
\end{eqnarray}
This completes the proof.
\end{proof}

\begin{prop}
The $n$-th derivative of the Dirac delta function, denoted by
$\delta^{(n)}$, is defined by the following:
\begin{eqnarray}
\Inner{\delta^{(n)}}{f} = (-1)^nf^{(n)}(0),
\end{eqnarray}
where $n\in\natural_+$ and $f$ is any function with continuous
derivatives at least up to the $n$-th order in some neighborhood of
$0$.
\end{prop}
From these, we see that
\begin{eqnarray}
\Inner{\delta'_a}{f} &=& -\Inner{\delta_a}{f'} = -f'(a),\\
\Inner{\delta^{(n)}_a}{f} &=& (-1)^n\Inner{\delta_a}{f^{(n)}} =
(-1)^nf^{(n)}(a).
\end{eqnarray}
Suppose that $g(t)$ increases monotonely over the closed interval
$[a,b]$: suppose there is $t_0\in(a,b)$ such that $g(t_0)=0$. We
have known that
$$
\frac{\dif g^{-1}(x)}{\dif x} = \frac1{\frac{\dif g(t)}{\dif t}}.
$$
From this, we see that, via $x=g(t), t\in[a,b]$,
\begin{eqnarray}
\int^b_a f(t)\delta(g(t))\dif t &=&
\int^{g(b)}_{g(a)}f(g^{-1}(x))\delta(x)\dif (g^{-1}(x))\\
&=&
\int^{g(b)}_{g(a)}f(g^{-1}(x))\delta(x)\frac{\dif g^{-1}(x)}{\dif x}\dif x\\
&=& f(g^{-1}(0))\frac{\dif g^{-1}(x)}{\dif x}\big|_{x=g(t_0)} =
\frac{f(t_0)}{g'(t_0)}.
\end{eqnarray}

\begin{prop}
If $g(t)$ is monotone, with $g(a)=0$ and $g'(a)\neq0$, then
\begin{eqnarray}
\delta(g(t)) = \frac{\delta_a(t)}{\abs{g'(a)}}.
\end{eqnarray}
\end{prop}
From this, we see that
$$
\delta(kx+b) = \frac1{\abs{k}}\delta\Pa{x+\frac bk},\quad k\neq0.
$$
More generally, the delta distribution may be composed with a smooth
function $g(x)$ in such a way that the familiar change of variables
formula holds, that
\begin{eqnarray}
\int_\real f(g(x))\delta(g(x))\abs{g'(x)}\dif x =
\int_{g(\real)}f(t)\delta(t)\dif t
\end{eqnarray}
provided that $g$ is a continuously differentiable function with
$g'$ nowhere zero. That is, there is a unique way to assign meaning
to the distribution $\delta\circ g$ so that this identity holds for
all compactly supported test functions $f$. Therefore, the domain
must be broken up to exclude the $g'(x) = 0$ point. This
distribution satisfies $\delta(g(x)) = 0$ if $g$ is nowhere zero,
and otherwise if $g$ has a real root at $x_0$, then
\begin{eqnarray}
\delta(g(x)) = \frac{\delta(x-x_0)}{\abs{g'(x_0)}}.
\end{eqnarray}
It is natural to define the composition $\delta(g(x))$ for
continuously differentiable functions $g$ by
\begin{eqnarray}
\delta(g(x)) = \sum_j \frac{\delta(x-x_j)}{\abs{g'(x_j)}}
\end{eqnarray}
where the sum extends over all roots of $g(x)$, which are assumed to
be simple. Thus, for example
\begin{eqnarray}
\delta\Pa{x^2-a^2} = \frac1{2\abs{a}}\Pa{\delta(x+a)+\delta(x-a)}.
\end{eqnarray}
In the integral form the generalized scaling property may be written
as
\begin{eqnarray}
\int^{+\infty}_{-\infty}f(x)\delta(g(x))\dif x = \sum_j
\frac{f(x_j)}{\abs{g'(x_j)}}.
\end{eqnarray}

\begin{exam}[\cite{Attila2017}]
If $T$ is an $n\times n$ positive definite matrix and
$r\in\real^+,\alpha\geqslant0$, then
\begin{eqnarray}
\int_{B_n(T,r)} \Pa{r - \Innerm{\bsu}{T}{\bsu}}^\alpha [\dif \bsu] =
\frac{\pi^n\Gamma(\alpha+1)}{\Gamma(n+\alpha+1)}\frac{r^{n+\alpha}}{\det(T)},
\end{eqnarray}
where $B_n(T,r):=\Set{\bsu\in\complex^n|\Innerm{\bsu}{T}{\bsu}<r}$
and $[\dif\bsu]=\prod^n_{j=1}\dif u_j$ for $[\dif
z]=\dif\Pa{\re(z)}\dif\Pa{\im(z)}$.

Indeed, let
$$
\bsv=r^{-\frac12}T^{\frac12}\bsu.
$$
Then $[\dif\bsv] = \det(r^{-1}T)[\dif\bsu]$ or $[\dif\bsu] =
\det(rT^{-1})[\dif\bsv]$. Thus
\begin{eqnarray}
\int_{B_n(T,r)} \Pa{r - \Innerm{\bsu}{T}{\bsu}}^\alpha [\dif \bsu] =
\frac{r^{n+\alpha}}{\det(T)}\int_{B_n(\I_n,1)} \Pa{1 -
\iinner{\bsv}{\bsv}}^\alpha[\dif\bsv],
\end{eqnarray}
where $B_n(\I_n,1)=\Set{\bsv\in\complex^n| \iinner{\bsv}{\bsv}<1}$.
Now
$$
B_n(\I_n,1)=\cup_{\gamma\in[0,1)}S_n(\gamma),
$$
where $S_n(\gamma)=\Set{\bsv\in\complex^n| \norm{\bsv}_2=\gamma}$.
\begin{eqnarray}
\int_{B_n(\I_n,1)} \Pa{1 - \iinner{\bsv}{\bsv}}^\alpha [\dif\bsv]
&=& \int^1_0\dif \gamma \int\delta(\gamma-\norm{\bsv}_2)\Pa{1 -
\iinner{\bsv}{\bsv}}^\alpha[\dif\bsv]\\
&=&\int^1_0 \dif\gamma
(1-\gamma^2)^\alpha\cdot \int\delta(\gamma-\norm{\bsv}_2)[\dif\bsv]\\
&=&\int^1_0 \dif\gamma (1-\gamma^2)^\alpha\vol(S_n(\gamma)),
\end{eqnarray}
where
$$
\vol(S_n(\gamma)) =\int\delta(\gamma-\norm{\bsv}_2)[\dif\bsv]=
\frac{2\pi^n}{\Gamma(n)}\gamma^{2n-1}.
$$
Therefore, we obtain that
\begin{eqnarray}
\int_{B_n(T,r)} \Pa{r - \Innerm{\bsu}{T}{\bsu}}^\alpha [\dif \bsu]
&=&\frac{2\pi^n}{\Gamma(n)}\frac{r^{n+\alpha}}{\det(T)}\int^1_0(1-\gamma^2)^\alpha
\gamma^{2n-1}\dif\gamma\\
&=&\frac{\pi^n}{\Gamma(n)}\frac{r^{n+\alpha}}{\det(T)}\int^1_0(1-\gamma^2)^\alpha
\gamma^{2n-2}\dif(\gamma^2),
\end{eqnarray}
where
$$
\int^1_0(1-\gamma^2)^\alpha \gamma^{2n-2}\dif(\gamma^2)
=\int^1_0(1-x)^\alpha x^{n-1}\dif x =
B(\alpha+1,n)=\frac{\Gamma(\alpha+1)\Gamma(n)}{\Gamma(n+\alpha+1)}.
$$
This completes the proof.
\end{exam}

\begin{definition}[Convolution]
An operation on functions, called \emph{convolution} and denoted by
the symbol $*$, is defined by:
\begin{eqnarray}
f*g(x) = \int_\real f(x-t)g(t)\dif t = \int_\real f(t)g(x-t)\dif t.
\end{eqnarray}
\end{definition}
It is easily seen that some properties of convolution:
\begin{enumerate}[(i)]
\item $\delta_a*\delta_b=\delta_{a+b}$.
\item The delta function as a convolution unit: $\delta* f=f*\delta=f$.
\item Convolution as the translation: $\delta_a*f=f*\delta_a=f_a$, where $f_a(x):=f(x-a)$.
\item $\delta^{(n)}*f= f*\delta^{(n)}=f^{(n)}$.
\end{enumerate}

\begin{definition}[Fourier transform]
Let $f$ be a complex-valued function of the real variable $t$ which
is absolutely integrable over the whole real axis $\real$. That is,
\begin{eqnarray*}
\int_\real \abs{f_1(x)}<+\infty\quad\text{and}\quad\int_\real
\abs{f_2(x)}<+\infty,
\end{eqnarray*}
where $f=f_1+\sqrt{-1}f_2$. We define the Fourier transform of $f$
to be a new function
\begin{eqnarray}
\widehat f(\omega) :=\cF(f)(\omega)=\int_\real e^{-\mathrm{i}\omega
t}f(t)\dif t.
\end{eqnarray}
\end{definition}

Next, we consider the Fourier integral representation of Dirac delta
function which is very powerful in applications. We use the
following standard result:
\begin{eqnarray}
\rP\rV\int^{+\infty}_{-\infty}\frac{e^{\mathrm{i}\omega x}}{x}\dif x
= \mathrm{i}\pi
\end{eqnarray}
where the symbol $\rP\rV$ denotes the Cauchy Principal Value of the
integral and $\omega>0$ a constant. That is,
\begin{eqnarray}
\int^{+\infty}_{-\infty}\frac{\sin(\omega x)}x \dif x=\pi
\quad\text{and}\quad \int^{+\infty}_{-\infty}\frac{\cos(\omega x)}x
\dif x=0.
\end{eqnarray}
Replacing $\omega$ by $-\omega$ simply changes the sign of the first
of these two real integrals and leaves the other unaltered. That is,
if $\omega>0$
\begin{eqnarray}
\rP\rV\int^{+\infty}_{-\infty}\frac{e^{-\mathrm{i}\omega x}}{x}\dif
x = -\mathrm{i}\pi
\end{eqnarray}
Hence, if we replace $\omega$ by the usual symbol $t$ for the
independent real variable we can write
\begin{eqnarray}
\rP\rV\int^{+\infty}_{-\infty}\frac{e^{\mathrm{i}t x}}{x}\dif x =
\mathrm{i}\pi\sign(t),
\end{eqnarray}
i.e.,
\begin{eqnarray}
\frac1{2\pi}\int^{+\infty}_{-\infty}\frac{e^{\mathrm{i}t
x}}{\mathrm{i}x}\dif x =\frac12\sign(t)=\begin{cases}
\frac12,&t>0\\-\frac12,&t<0\end{cases}
\end{eqnarray}
A formal differentiation of this with respect to $t$ then yields the
following result:
\begin{prop}
It holds that
\begin{eqnarray}\label{eq:inverse-of-delta}
\delta(t)=\frac1{2\pi}\int^{+\infty}_{-\infty}e^{\mathrm{i}t x}\dif
x.
\end{eqnarray}
\end{prop}
This amounts to say $\cF^{-1}(1)(t)=\delta(t)$. Replacing $t$ by
$t-a$, we have
\begin{eqnarray}\label{eq:inverse-of-delta-a}
\frac1{2\pi}\int^{+\infty}_{-\infty}e^{\mathrm{i}(t-a) x}\dif x =
\delta(t-a)=\delta_a(t).
\end{eqnarray}
This amounts to say $\cF^{-1}(e^{-\mathrm{i}ax})(t)=\delta_a(t)$ or
$\cF(\delta_a(t))(x)=e^{-\mathrm{i}ax}$. The integral on the
left-hand side of \eqref{eq:inverse-of-delta} is, of course,
divergent, and it is clear that this equation must be understood
symbolically. That is to say, for all sufficiently well-behaved
functions $f$, we should interpret \eqref{eq:inverse-of-delta} to
mean that
\begin{eqnarray}
\int^{+\infty}_{-\infty}f(t)
\Br{\frac1{2\pi}\int^{+\infty}_{-\infty}e^{\mathrm{i}t \omega}\dif
\omega}\dif t = \int^{+\infty}_{-\infty}f(t)\delta(t)\dif t = f(0)
\end{eqnarray}
or, more generally, that
\begin{eqnarray}
\int^{+\infty}_{-\infty}f(x)
\Br{\frac1{2\pi}\int^{+\infty}_{-\infty}e^{\mathrm{i}(t-x)
\omega}\dif \omega}\dif x =
\int^{+\infty}_{-\infty}f(x)\delta(t-x)\dif x = f(t).
\end{eqnarray}
We can rewrite this result in the form
\begin{eqnarray}
f(t) &=& \int^{+\infty}_{-\infty}f(x)
\Br{\frac1{2\pi}\int^{+\infty}_{-\infty}e^{\mathrm{i}(t-x)
\omega}\dif \omega}\dif x\\
&=&\frac1{2\pi}\int^{+\infty}_{-\infty}e^{\mathrm{i}t\omega}\Br{\int^{+\infty}_{-\infty}
f(x)e^{-\mathrm{i}x\omega}\dif x}\dif \omega =
\frac1{2\pi}\int^{+\infty}_{-\infty}e^{\mathrm{i}t\omega}\widehat
f(\omega)\dif \omega
\end{eqnarray}

\begin{prop}[Fourier Inversion]
Let $f$ be a (real or complex valued) function of a single real
variable which is absolutely integrable over the interval
$(-\infty,+\infty)$ and which also satisfies the Dirichlet
conditions over every finite interval. If $\widehat f(\omega)$
denotes the Fourier transform of $f$, then at each point $t$ we have
\begin{eqnarray}
\frac1{2\pi}\int^{+\infty}_{-\infty}e^{\mathrm{i}t\omega}\widehat
f(\omega)\dif \omega = \frac12\Br{f(t+)+f(t-)},
\end{eqnarray}
where $f(t\pm):=\lim_{s\to t^\pm}f(s)$.
\end{prop}
There are several important properties of the Fourier transform
which merit explicit mention.
\begin{enumerate}[(i)]
\item The Fourier transform of the convolution of two functions is equal to the product of their individual transforms:
$\cF(f*g)=\cF(f)\cF(g)$.
\item $\cF(fg) =
\frac1{2\pi}\cF(f)*\cF(g)$.
\item The Fourier transformation is
linear: $\cF(\lambda_1f_1+\lambda_2f_2) = \lambda_1\cF(f_1) +
\lambda_2\cF(f_2)$.
\item $\cF(f(x-a))(\omega) = e^{-\mathrm{i}\omega
a}\cF(f(x))(\omega)$.
\item $\cF(f(x)e^{-ax})(\omega) = \cF(f(x))(a+\mathrm{i}\omega)$.
\item $\cF(f')(\omega) = \mathrm{i}\omega\cF(f)(\omega)$.
\item $\cF(f(ax))(\omega)=\frac1a\cF(f(x))\Pa{\frac{\omega}a}$.
\end{enumerate}
For example, the proofs of (iv) and (vi) are given. Indeed,
$f(x-a)=f_a(x)=f*\delta_a(x)$, thus $\cF(f(x-a))(\omega) =
\cF(f*\delta_a(x))(\omega) = \cF(f)\cF(\delta_a)(\omega)$, that is,
$\cF(f(x-a))(\omega)=e^{-\mathrm{i}\omega a}\cF(f)$, hence (iv).
Since $f'=f*\delta'$, it follows that
$\cF(f')=\cF(f*\delta')=\cF(f)\cF(\delta')$. In what follows, we
calculate $\cF(\delta')$. By definition of Fourier transform,
\begin{eqnarray}
\cF(\delta')(\omega) = \int_\real e^{-\mathrm{i}\omega
t}\delta'(t)\dif t = -\frac{\dif e^{-\mathrm{i}\omega t}}{\dif
t}\big|_{t=0} =\mathrm{i}\omega.
\end{eqnarray}
Thus $\cF(f')(\omega) = \mathrm{i}\omega\cF(f)(\omega)$, hence (vi).
This property can be generalized:
$\cF(f^{(n)})(\omega)=(\mathrm{i}\omega)^n\cF(f)(\omega)$. Indeed,
\begin{eqnarray}
\cF(f^{(n)})(\omega)=\cF(f*\delta^{(n)})(\omega)=\cF(f)(\omega)\cF(\delta^{(n)})(\omega)=(\mathrm{i}\omega)^n\cF(f)(\omega).
\end{eqnarray}

We can apply the sampling property of the delta function to the
Fourier inversion integral:
\begin{eqnarray}
\frac1{2\pi}\int^{+\infty}_{-\infty}e^{\mathrm{i}x\omega}\delta(\omega-\alpha)\dif
\omega=\frac1{2\pi}e^{\mathrm{i}x\alpha}
\end{eqnarray}
and similarly
\begin{eqnarray}
\frac1{2\pi}\int^{+\infty}_{-\infty}e^{\mathrm{i}x\omega}\delta(\omega+\alpha)\dif
\omega=\frac1{2\pi}e^{-\mathrm{i}x\alpha}.
\end{eqnarray}
Thus, recalling that the Fourier transform is defined in general for
complex-valued functions, these results suggest that we can give the
following definitions for the Fourier transforms of complex
exponentials such as
$$
\cF(e^{\mathrm{i}\alpha x})(\omega) =
2\pi\delta(\omega-\alpha);\quad \cF(e^{-\mathrm{i}\alpha x})(\omega)
= 2\pi\delta(\omega+\alpha).
$$
Both equations immediately yield the following definitions for the
Fourier transforms of the real functions $\cos(\alpha x)$ and
$\sin(\alpha x)$:
\begin{eqnarray}
\cF(\cos(\alpha x))(\omega) = \pi\Pa{\delta(\omega-\alpha) +
\delta(\omega+\alpha)},\\
\cF(\sin(\alpha x))(\omega) =
-\mathrm{i}\pi\Pa{\delta(\omega-\alpha) - \delta(\omega+\alpha)}.
\end{eqnarray}
In particular, taking $\alpha = 0$, we find that the generalized
Fourier transform of the constant function $f(t) \equiv 1$ is simply
$2\pi\delta(\omega)$. This in turn allows us to offer a definition
of the Fourier transform of the unit step function.
\begin{prop}
The Fourier transform of the Heaviside step function
$H(x)=\frac12+\frac12\sign(x)$, where $\sign(x)=\frac{x}{\abs{x}}$,
is given by
\begin{eqnarray}
\widehat H(\omega) = \pi\delta(\omega) + \frac1{\mathrm{i}\omega}.
\end{eqnarray}
\end{prop}

\begin{proof}
Now
\begin{eqnarray}
\frac1{2\pi}\int^{+\infty}_{-\infty}\frac{e^{\mathrm{i}x
\omega}}{\mathrm{i}\omega}\dif \omega =\frac12\sign(x)=\begin{cases}
\frac12,&x>0\\-\frac12,&x<0\end{cases}
\end{eqnarray}
Then we know that $\frac2{\mathrm{i}\omega}$ is a suitable choice
for the Fourier transform of the function $\sign(x)$ in the sense
that
\begin{eqnarray}
\sign(x) = \frac1{2\pi}\int^{+\infty}_{-\infty}e^{\mathrm{i}x
\omega}\frac2{\mathrm{i}\omega}\dif \omega =
\cF^{-1}\Pa{\frac2{\mathrm{i}\omega}}(x).
\end{eqnarray}
This amounts to say that $\cF(\sign(x))(\omega) =
\frac2{\mathrm{i}\omega}$. Hence for Heaviside step function
\begin{eqnarray}
H(x) = \frac12+\frac12\sign(x),
\end{eqnarray}
the Fourier transform of it is given by
\begin{eqnarray}
\cF(H(x))(\omega) = \cF\Pa{\frac12+\frac12\sign(x)} =
\frac12\cF(1)(\omega) + \frac12\cF(\sign(x))(\omega),
\end{eqnarray}
i.e.
\begin{eqnarray}
\cF(H(x))(\omega) = \pi\delta(\omega) +
\frac1{\mathrm{i}\omega}\Longleftrightarrow \widehat H(\omega) =
\pi\delta(\omega) + \frac1{\mathrm{i}\omega}.
\end{eqnarray}
This completes the proof.
\end{proof}
We can get some important properties of Dirac delta function which
are listed below:
\begin{enumerate}[(i)]
\item The delta function is an even distribution: $\delta(x) = \delta(-x)$.
\item The delta function satisfies the following scaling property for a non-zero scalar: $\delta(ax) = \frac1{\abs{a}}\delta(x)$ for $a\in\real\backslash\set{0}$.
\item The distributional product of $\delta(x)$ and $x$ is equal to
zero: $x\delta(x)=0$.
\item If $xf(x)=xg(x)$, where $f$ and $g$ are distributions, then
$f(x)=g(x)+c\delta(x)$ for some constant $c$.
\end{enumerate}
Previous two facts can be checked as follows: Note that
\begin{eqnarray*}
\delta(-x) &=&
\frac1{2\pi}\int^\infty_{-\infty}e^{-\mathrm{i}tx}\dif t
=-\frac1{2\pi}\int^\infty_{-\infty}e^{-\mathrm{i}tx}\dif (-t) \\
&=& -\frac1{2\pi}\int^{-\infty}_{+\infty}e^{\mathrm{i}sx}\dif s =
\frac1{2\pi}\int^\infty_{-\infty}e^{\mathrm{i}sx}\dif s = \delta(x).
\end{eqnarray*}
This is (i). For the proof of (ii), since $a\neq0$, we observe that
$\delta(ax)=\delta(-ax)$ by (i), hence
$\delta(ax)=\delta(\abs{a}x)$, it follows that
\begin{eqnarray*}
\delta(ax) = \delta(\abs{a}x) =
\frac1{2\pi}\int^{+\infty}_{-\infty}e^{\mathrm{i}t\cdot
\abs{a}x}\dif t= \frac1{2\pi}\int^{+\infty}_{-\infty}e^{\mathrm{i}
\abs{a}t\cdot x}\dif t.
\end{eqnarray*}
Let $s=\abs{a}t$. Then $\dif s=\abs{a}\dif t$, thus
\begin{eqnarray}
\delta(ax) =
\frac1{\abs{a}}\frac1{2\pi}\int^{+\infty}_{-\infty}e^{\mathrm{i}
s\cdot x}\dif s = \frac1{\abs{a}}\delta(x).
\end{eqnarray}
That is, in the sense of distribution,
\begin{eqnarray}
\delta(ax) = \frac1{\abs{a}}\delta(x),\quad
a\in\real\backslash\set{0}.
\end{eqnarray}

\section{Dirac delta function of vector argument}

\begin{definition}[Dirac delta function of real-vector arguments]
The real-vector delta function can be defined in $n$-dimensional
Euclidean space $\real^n$ as the measure such that
\begin{eqnarray}
\int_{\real^n} f(\bsx)\delta(\bsx)[\dif \bsx] = f(\mathbf{0})
\end{eqnarray}
for every compactly supported continuous function $f$. As a measure,
the $n$-dimensional delta function is the product measure of the
1-dimensional delta functions in each variable separately. Thus,
formally, with
\begin{eqnarray}
\delta(\bsx)=\prod^n_{j=1}\delta(x_j),
\end{eqnarray}
where $\bsx=[x_1,\ldots,x_n]^\t\in\real^n$.
\end{definition}
The delta function in an $n$-dimensional space satisfies the
following scaling property instead:
\begin{eqnarray}\label{eq:scalar-product}
\delta(a\bsx)=\abs{a}^{-n}\delta(\bsx),\quad
a\in\real\backslash\set{0}.
\end{eqnarray}
Indeed, $a\bsx = \Br{ax_1,\ldots, ax_n}^\t$ for
$\bsx=\Br{x_1,\ldots,x_n}^\t$, thus
$$
\delta(a\bsx) = \prod^n_{j=1}\delta(ax_j) =
\prod^n_{j=1}\abs{a}^{-1}\delta(x_j) =
\abs{a}^{-n}\prod^n_{j=1}\delta(x_j)=\abs{a}^{-n}\delta(\bsx).
$$
This indicates that $\delta$ is a homogeneous distribtion of degree
$(-n)$. As in the one-variable case, it is possible to define the
compositon of $\delta$ with a bi-Lipschitz function
$g:\real^n\to\real^n$ uniquely so that the identity
\begin{eqnarray}
\int_{\real^n}f(g(\bsx))\delta(g(\bsx))\abs{\det g'(\bsx)}[\dif
\bsx] = \int_{g(\real^n)}f(\bsu)\delta(\bsu)[\dif \bsu]
\end{eqnarray}
for all compactly supported functions $f$.

Using the coarea formula from geometric measure theory, one can also
define the composition of the delta function with a submersion from
one Euclidean space to another one of different dimension; the
result is a type of current. In the special case of a continuously
differentiable function $g: \real^n\to\real$ such that the gradient
of $g$ is nowhere zero, the following identity holds\footnote{See
\url{https://en.wikipedia.org/wiki/Dirac_delta_function}}
\begin{eqnarray}
\int_{\real^n}f(\bsx)\delta(g(\bsx))[\dif\bsx] =
\int_{g^{-1}(0)}\frac{f(\bsx)}{\abs{\nabla g(\bsx)}}\dif\sigma(\bsx)
\end{eqnarray}
where the integral on the right is over $g^{-1}(0)$, the
$(n-1)$-dimensional surface defined by $g(\bsx)=0$ with respect to
the Minkowski content  measure. That is known as a simple layer
integral.

\begin{prop}
It holds that
\begin{eqnarray}
\delta(\bsx) = \frac1{(2\pi)^n} \int_{\real^n}
e^{\mathrm{i}\Inner{\bst}{\bsx}}[\dif \bst]\quad (\bsx\in\real^n).
\end{eqnarray}
\end{prop}

\begin{proof}
Let $\bsx=[x_1,\ldots,x_n]^\t\in\real^n$. Then by Fourier transform
of Dirac delta function:
\begin{eqnarray}
\delta(\bsx) &=& \prod^n_{j=1}\delta(x_j) =
\prod^n_{j=1}\frac1{2\pi}\int_\real e^{\mathrm{i}t_jx_j}\dif t_j\\
&=& \frac1{(2\pi)^n}
\int_{\real^n}e^{\mathrm{i}\sum^n_{j=1}t_jx_j}\prod^n_{j=1}\dif
t_j\\
&=&\frac1{(2\pi)^n} \int_{\real^n}
e^{\mathrm{i}\Inner{\bst}{\bsx}}[\dif \bst]\quad (\bsx\in\real^n),
\end{eqnarray}
where $[\dif\bst]:=\prod^n_{j=1}\dif t_j$.
\end{proof}

\begin{prop}\label{prop:transformation-delta}
For a full-ranked real matrix $A\in\real^{n\times n}$, it holds that
\begin{eqnarray}
\delta(A\bsx) = \frac1{\abs{\det(A)}}\delta(\bsx),\quad
\bsx\in\real^n.
\end{eqnarray}
In particular, under any reflection or rotation $R$, the delta
function is invariant:
\begin{eqnarray}
\delta(R\bsx) = \delta(\bsx).
\end{eqnarray}
\end{prop}

\begin{proof}[The first proof]
By using Fourier transform of Dirac delta function, it follows that
\begin{eqnarray}
\delta(A\bsx) = \frac1{(2\pi)^n}\int_{\real^n}
e^{\mathrm{i}\Inner{\bst}{A\bsx}}[\dif \bst]=
\frac1{(2\pi)^n}\int_{\real^n}
e^{\mathrm{i}\Inner{A^\t\bst}{\bsx}}[\dif \bst]\quad
(\bsx\in\real^n).
\end{eqnarray}
Let $\bss=A^\t\bst$. Then
$[\dif\bss]=\abs{\det(A^\t)}[\dif\bst]=\abs{\det(A)}[\dif\bst]$.
From this, we see that
\begin{eqnarray}
\delta(A\mathbf{x}) =
\abs{\det^{-1}(A)}\frac1{(2\pi)^n}\int_{\real^n}
e^{\mathrm{i}\Inner{\bss}{\bsx}}[\dif \bss] =
\abs{\det^{-1}(A)}\delta(\bsx) \quad (\bsx\in\real^n).
\end{eqnarray}
Thus $\delta(R\bsx) = \delta(\bsx)$ since $\det(R)=\pm1$ for
reflection or any rotation $R$. We are done.
\end{proof}

\begin{proof}[The second proof]
By SVD, we have two orthogonal matrices $L,R$ and diagonal matrix
$$
\Lambda=\diag(\lambda_1,\ldots,\lambda_n)
$$
with positive diagonal
entries such that $A=L\Lambda R^\t$. Then, via $\bsy:= R^\t\bsx$
(hence $\delta(\bsy) =\delta(\bsx)$),
$$
\delta(A\bsx) = \delta(L\Lambda R^\t\bsx) = \delta(\Lambda
\bsy)=\prod^n_{j=1}\delta(\lambda_jy_j)=\prod^n_{j=1}\lambda^{-1}_j\delta(y_j),
$$
that is,
$$
\delta(A\bsx) =
\frac1{\prod^n_{j=1}\lambda_j}\prod^n_{j=1}\delta(y_j)=
\frac1{\prod^n_{j=1}\lambda_j}\delta(\bsy) =
\frac1{\prod^n_{j=1}\lambda_j}\delta(\bsx).
$$
Now $\det(A)=\det(L\Lambda
R^\t\bsx)=\det(L)\det(\Lambda)\det(R^\t)$, it follows that
$$
\abs{\det(A)} =
\abs{\det(L)}\abs{\det(\Lambda)}\abs{\det(R^\t)}=\prod^n_{j=1}\lambda_j.
$$
Therefore we get the desired identity:
$\delta(A\bsx)=\abs{\det^{-1}(A)}\delta(\bsx)$. This completes the
proof.
\end{proof}
Clearly, letting $A=a\I_n$ in the above gives
Eq.~\eqref{eq:scalar-product}.

\section{Dirac delta function of matrix argument}

\begin{definition}[Dirac delta function of real-matrix argument]
(i) For an $m\times n$ real matrix $X=[x_{ij}]\in\real^{m\times n}$, the
matrix delta function $\delta(X)$ is defined as
\begin{eqnarray}
\delta(X):=\prod^m_{i=1}\prod^n_{j=1}\delta(x_{ij}).
\end{eqnarray}
In particular, the vector delta function is just a special case
where
$n=1$ in the matrix case.\\
(ii) For an $m\times m$ symmetric real matrix $X=[x_{ij}]$, the
matrix delta function $\delta(X)$ is defined as
\begin{eqnarray}
\delta(X):=\prod_{i\leqslant j}\delta(x_{ij}).
\end{eqnarray}
\end{definition}
From the above definition, we see that the matrix delta function of
a complex matrix is equal to the product of one-dimensional delta
functions over the independent real and imaginary parts of this
complex matrix. In view of this observation, we see that $\delta(X)
= \delta(\col{X})$, where $\col{X}$ is the vectorization of the
matrix $X$. It is easily checked for a rectangular matrix. For the
symmetric case, for example, take $2\times 2$ symmetric real matrix
$X=\Br{\begin{array}{cc}
         x_{11} & x_{12} \\
         x_{21} & x_{22}
       \end{array}
}$ with $x_{12}=x_{21}$, then $\col{X} = [x_{11},x_{21},x_{12},
x_{22}]^\t=[x_{11},x_{12},x_{12}, x_{22}]^\t$, thus
$\col{X}=x_{11}[1,0,0,0]^\t+x_{12}[0,1,1,0]^\t+x_{22}[0,0,0,1]^\t$,
i.e., there are three independent variables
$\set{x_{11},x_{12},x_{22}}$ in the vector $\col{X}$ just like in
the matrix $X$, thus
$$
\delta(X) =
\delta(x_{11})\delta(x_{12})\delta(x_{22})=\delta(\col{X}).
$$
\begin{prop}
For an $m\times m$ symmetric matrix $X$, we have
\begin{eqnarray}\label{eq:real-matrix-delta-function}
\delta(X) = 2^{-m}\pi^{-\frac{m(m+1)}2}\int
e^{\mathrm{i}\Tr{TX}}[\dif T],
\end{eqnarray}
where $T=[t_{ij}]$ is also an $m\times m$ real symmetric matrix, and
$[\dif T]:=\prod_{i\leqslant j}\dif t_{ij}$.
\end{prop}

\begin{proof}
Since
\begin{eqnarray*}
\Tr{TX} &=& \sum_{j=1}^m t_{jj}x_{jj} + \sum_{i\neq j}t_{ij}x_{ij}=
\sum_{j=1}^m t_{jj}x_{jj} + 2\sum_{i<j}t_{ij}x_{ij}
\end{eqnarray*}
implying that
\begin{eqnarray*}
\int e^{\mathrm{i}\Tr{TX}}[\dif T] &=& \prod_{j=1}^m \int
\exp\Pa{\mathrm{i}t_{jj}x_{jj}}\dif t_{jj}\prod_{1\leqslant i<
j\leqslant m}\int \exp\Pa{\mathrm{i}t_{ij}\Pa{2x_{ij}}}\dif
t_{ij}\\
&=& \prod^m_{j=1}2\pi\delta\Pa{x_{jj}}\times \prod_{1\leqslant
i<j\leqslant m}2\pi\delta\Pa{2x_{ij}}\\
&=&\prod^m_{j=1}2\pi\delta\Pa{x_{jj}}\times \prod_{1\leqslant
i<j\leqslant m}\pi\delta(x_{ij}) = 2^m\pi^{\binom{m+1}{2}}\delta(X).
\end{eqnarray*}
Therefore we get the desired identity.
\end{proof}

\begin{prop}
For $A\in\real^{m\times m},B\in\real^{n\times n}$ and
$X\in\real^{m\times n}$, we have
\begin{eqnarray}\label{eq:AXB-delta}
\delta(AXB) = \abs{\det^{-n}(A)\det^{-m}(B)}\delta(X).
\end{eqnarray}
\end{prop}

\begin{proof}
We have already known that $\delta(AXB) = \delta(\col{AXB})$. Since
$\col{AXB} = (A\ot B^\t)\col{X}$ \cite{Zhang2015vol}, it follows
that
\begin{eqnarray}
\delta(AXB) &=& \delta(\col{AXB}) = \delta\Pa{(A\ot B^\t)\col{X}}\\
&=&\abs{\det^{-1}(A\ot B^\t)}\delta(\col{X}) \\
&=&\abs{\det^{-n}(A)\det^{-m}(B)}\delta(X).
\end{eqnarray}
This completes the proof.
\end{proof}

\begin{prop}
For $A\in\real^{n\times n}$ and $X=X^\t\in\real^{n\times n}$, we
have
\begin{eqnarray}
\delta(AXA^\t) = \abs{\det(A)}^{-(n+1)}\delta(X).
\end{eqnarray}
\end{prop}

\begin{proof}
By using Eq.~\eqref{eq:real-matrix-delta-function}, it follows that
\begin{eqnarray*}
\delta(AXA^\t) &=& 2^{-n}\pi^{-\frac{n(n+1)}2}\int
e^{\mathrm{i}\Tr{TAXA^\t}}[\dif T]\\
&=&2^{-n}\pi^{-\frac{n(n+1)}2}\int e^{\mathrm{i}\Tr{A^\t TAX}}[\dif
T].
\end{eqnarray*}
Let $S=A^\t TA$. Then we have $[\dif S]=\abs{\det(A)}^{n+1}[\dif T]$
(see Proposition 2.8 in \cite{Zhang2015vol}). From this, we see that
\begin{eqnarray}
\delta(AXA^\t) =\abs{\det(A)}^{-(n+1)}
2^{-n}\pi^{-\frac{n(n+1)}2}\int e^{\mathrm{i}\Tr{SX}}[\dif S] =
\abs{\det(A)}^{-(n+1)}\delta(X).
\end{eqnarray}
This completes the proof.
\end{proof}

Note that Dirac delta function for complex number is defined by
$\delta(z):=\delta(\re(z))\delta(\im(z))$, where
$z=\re(z)+\sqrt{-1}\im(z)$ for $\re(z),\im(z)\in\real$. The complex
number $z$ can be realized as a 2-dimensional real vector
$$
z\mapsto \widehat z:=\Br{\begin{array}{c}
               \re(z) \\
               \im(z)
             \end{array}
}.
$$
Thus
\begin{eqnarray}
\delta(z) = \delta(\widehat z)=\delta\Pa{\Br{\begin{array}{c}
               \re(z) \\
               \im(z)
             \end{array}
}}.
\end{eqnarray}
Then for $c\in\complex$, $cz$ is represented as
$$
cz \mapsto \Br{\begin{array}{c}
               \re(c)\re(z)-\im(c)\im(z) \\
                \im(c)\re(z)+\re(c)\im(z)
             \end{array}
} = \Br{\begin{array}{cc}
          \re(c) & -\im(c) \\
          \im(c) & \re(c)
        \end{array}
}\Br{\begin{array}{c}
               \re(z) \\
               \im(z)
             \end{array}
},
$$
we have
\begin{eqnarray*}
\delta(cz) &=& \delta\Pa{\Br{\begin{array}{cc}
          \re(c) & -\im(c) \\
          \im(c) & \re(c)
        \end{array}
}\Br{\begin{array}{c}
               \re(z) \\
               \im(z)
             \end{array}
}}\\
&=&\abs{\det^{-1}\Pa{\Br{\begin{array}{cc}
          \re(c) & -\im(c) \\
          \im(c) & \re(c)
        \end{array}
}}}\delta\Pa{\Br{\begin{array}{c}
               \re(z) \\
               \im(z)
             \end{array}
}} = \abs{c}^{-2}\delta(z).
\end{eqnarray*}
Therefore we have
\begin{eqnarray}
\delta(cz)=\abs{c}^{-2}\delta(z).
\end{eqnarray}
Furthermore, if $\bsz\in\complex^n$, then
\begin{eqnarray}
\delta(c\bsz)=\abs{c}^{-2n}\delta(\bsz).
\end{eqnarray}

\begin{prop}
For a full-ranked complex matrix $A\in\complex^{n\times n}$, it
holds that
\begin{eqnarray}
\delta(A\bsz) =
\frac1{\abs{\det(AA^*)}}\delta(\bsz)=\abs{\det(A)}^{-2}\delta(\bsz),\quad
\bsz\in\complex^n.
\end{eqnarray}
In particular, for any unitary matrix $U\in\rU(n)$, we have
\begin{eqnarray}
\delta(U\bsz)=\delta(\bsz).
\end{eqnarray}
\end{prop}
\begin{proof}
Since $A\bsz$ can be represented as
$$
\widehat A\widehat{\bsz}=\Br{\begin{array}{cc}
      \re(A) & -\im(A) \\
      \im(A) & \re(A)
    \end{array}
}\Br{\begin{array}{c}
               \re(\bsz) \\
               \im(\bsz)
             \end{array}
}
$$
it follows that
\begin{eqnarray*}
\delta(A\bsz) &=& \delta\Pa{\Br{\begin{array}{cc}
      \re(A) & -\im(A) \\
      \im(A) & \re(A)
    \end{array}
}\Br{\begin{array}{c}
               \re(\bsz) \\
               \im(\bsz)
             \end{array}
}}\\
&=& \abs{\det^{-1}\Pa{\Br{\begin{array}{cc}
      \re(A) & -\im(A) \\
      \im(A) & \re(A)
    \end{array}
}}}\delta\Pa{\Br{\begin{array}{c}
               \re(\bsz) \\
               \im(\bsz)
             \end{array}
}}\\
&=& \abs{\det^{-1}(AA^*)}\delta(\bsz),
\end{eqnarray*}
where $\abs{\det\Pa{\Br{\begin{array}{cc}
      \re(A) & -\im(A) \\
      \im(A) & \re(A)
    \end{array}
}}}=\abs{\det(AA^*)}$ can be found in
\cite{Mathai1997,Zhang2015vol}. Therefore we have
\begin{eqnarray}
\delta(A\bsz) = \abs{\det(A)}^{-2}\delta(\bsz),\quad
\bsz\in\complex^n.
\end{eqnarray}
If $A=U$ is a unitary matrix, then $\abs{\det(U)}=1$. The desired
result is obtained.
\end{proof}

\begin{proof}[The second proof]
Let $\widehat A = \Br{\begin{array}{cc}
      \re(A) & -\im(A) \\
      \im(A) & \re(A)
    \end{array}
}$ and $\widehat{\bsz} = \Br{\begin{array}{c}
               \re(\bsz) \\
               \im(\bsz)
             \end{array}
}$. Then by Proposition~\ref{prop:transformation-delta},
\begin{eqnarray}
\delta(A\bsz) = \delta(\widehat A \widehat{\bsz}) =
\abs{\det^{-1}(\widehat A)}\delta(\widehat{\bsz}) =
\abs{\det^{-1}(AA^*)}\delta(\bsz).
\end{eqnarray}
This completes the proof.
\end{proof}

\begin{definition}[Dirac delta
function of complex-matrix argument] (i) For an $m\times n$ complex
matrix $Z=[z_{ij}]\in\complex^{m\times n}$, the matrix delta
function $\delta(Z)$ is defined as
\begin{eqnarray}
\delta(Z):=\prod^m_{i=1}\prod^n_{j=1}\delta\Pa{\re(z_{ij})}\delta\Pa{\im(z_{ij})}.
\end{eqnarray}
In particular, the vector delta function is just a special case
where
$n=1$ in the matrix case.\\
(ii) For an $m\times m$ Hermitian complex matrix
$X=[x_{ij}]\in\complex^{m\times m}$, the matrix delta function
$\delta(X)$ is defined as
\begin{eqnarray}
\delta(X):=\prod_j \delta(x_{jj})\prod_{i<
j}\delta\Pa{\re(x_{ij})}\delta\Pa{\im(x_{ij})}.
\end{eqnarray}
\end{definition}
The Fourier integral representation of Dirac delta function can be
extended to the matrix case. The following proposition is very
important in this paper.
\begin{prop}\label{prop:Fourier-Matrix}
For an $m\times m$ Hermitian complex matrix $X\in\complex^{m\times
m}$, we have
\begin{eqnarray}\label{eq:matrix-delta-function}
\delta(X) = \frac1{2^m\pi^{m^2}}\int e^{\mathrm{i}\Tr{TX}}[\dif T],
\end{eqnarray}
where $T=[t_{ij}]$ is also an $m\times m$ Hermitian complex matrix,
and $[\dif T]:=\prod_j\dif t_{jj}\prod_{i<j}\dif \re(t_{ij})\dif
\im(t_{ij})$.
\end{prop}

\begin{proof}
Indeed, we know that
\begin{eqnarray*}
\Tr{TX} &=& \sum_{j=1}^m t_{jj}x_{jj} + \sum_{i\neq j}\Pa{\bar
t_{ij}x_{ij}}= \sum_{j=1}^m \re(t_{jj})\re(x_{jj}) +
\sum_{1\leqslant i<
j\leqslant m}\Pa{\bar t_{ij}x_{ij} + t_{ij}\bar x_{ij}}\\
&=&\sum_{j=1}^m t_{jj}x_{jj} + \sum_{1\leqslant i< j\leqslant
m}2\Pa{\re(t_{ij})\re(x_{ij}) + \im(t_{ij})\im(x_{ij})},
\end{eqnarray*}
implying that
\begin{eqnarray*}
\int e^{\mathrm{i}\Tr{TX}}[\dif T] &=& \prod_{j=1}^m \int
\exp\Pa{\mathrm{i}t_{jj}x_{jj}}\dif t_{jj}\\
&&\times\prod_{1\leqslant i< j\leqslant m}\int
\exp\Pa{\mathrm{i}\re(t_{ij})\Pa{2\re(x_{ij})}}\dif
\re(t_{ij})\\
&&\times\prod_{1\leqslant i< j\leqslant
m}\int\exp\Pa{\mathrm{i}\im(t_{ij}) \Pa{2\im(x_{ij})}}\dif
\im(t_{ij})\\
&=&  \prod^m_{j=1}2\pi\delta\Pa{x_{jj}}\times \prod_{1\leqslant
i<j\leqslant m}2\pi\delta\Pa{2\re(x_{ij})}2\pi\delta\Pa{2\im(x_{ij})}\\
&=& \prod^m_{j=1}2\pi\delta\Pa{x_{jj}}\times \prod_{1\leqslant
i<j\leqslant m}\pi\delta\Pa{\re(x_{ij})}\pi\delta\Pa{\im(x_{ij})}\\
&=& \Pa{2\pi}^m\Pa{\pi^2}^{\binom{m}{2}}\prod_j\delta(x_{jj})\prod_{i< j}\delta\Pa{\re(x_{ij})}\delta\Pa{\im(x_{ij})}\\
&=& 2^m\pi^{m^2}\delta(X).
\end{eqnarray*}
Therefore we get the desired identity.
\end{proof}

\begin{remark}
Indeed, since
\begin{eqnarray*}
\Tr{T^{\mathrm{off}}X^{\mathrm{off}}} =
\sum_{i<j}2(\re(t_{ij})\re(x_{ij}) + \im(t_{ij})\im(x_{ij}))
\end{eqnarray*}
and
\begin{eqnarray*}
[\dif T^{\mathrm{off}}] = \prod_{i<j}\dif\re(t_{ij})\dif
\im(t_{ij}),
\end{eqnarray*}
it follows that
\begin{eqnarray*}
&&\int[\dif T^{\mathrm{off}}]\exp\Pa{\mathrm{i}\Tr{T^{\mathrm{off}}X^{\mathrm{off}}}}\\
&&=
\prod_{i<j}\int\dif\re(t_{ij})\exp\Pa{\mathrm{i}\re(t_{ij})(2\re(x_{ij}))}\int\dif\im(t_{ij})\exp\Pa{\mathrm{i}\im(t_{ij})(2\im(x_{ij}))}\\
&&=\prod_{i<j}2\pi\delta(2\re(x_{ij}))\cdot2\pi\delta(2\im(x_{ij}))
=
\prod_{i<j}\pi\delta(\re(x_{ij}))\pi\delta(\im(x_{ij}))\\
&&=\pi^{2\binom{m}{2}}\prod_{i<j}\delta(\re(x_{ij}))\delta(\im(x_{ij}))
= \pi^{m(m-1)}\delta(X^{\mathrm{off}}).
\end{eqnarray*}
From the above discussion, we see that
\eqref{eq:matrix-delta-function} can be separated into two
identities below:
\begin{eqnarray}
\delta(X^\diag) &=& \frac1{(2\pi)^m}\int[\dif
T^\diag]e^{\mathrm{i}\Tr{T^\diag X^\diag}},\\
\delta(X^\mathrm{off}) &=& \frac1{\pi^{m(m-1)}}\int[\dif
T^\mathrm{off}]e^{\mathrm{i}\Tr{T^\mathrm{off}X^\mathrm{off}}}.
\end{eqnarray}
Note that the identity in Proposition~\ref{prop:Fourier-Matrix} is
used in deriving the joint distribution of diagonal part of Wishart
matrix ensemble \cite{Zhang2016}, and it is also used in obtaining
derivative principle for unitarily invariant random matrix ensemble
in \cite{Mejia2017}. More generally, the fact can be found in
\cite{Christandl2014} that the derivative principle for invariant
measure is used to investigate the joint distribution of eigenvalues
of local states from the same multipartite pure states.
\end{remark}

\begin{prop}
For $A\in\complex^{m\times m}$ and  $B\in\complex^{n\times n}$, let
$Z\in\complex^{m\times n}$. Then we have
\begin{eqnarray}
\delta(AZB) = \det^{-n}(AA^*)\det^{-m}(BB^*)\delta(Z).
\end{eqnarray}
\end{prop}

\begin{proof}
Now $AZB$ can be represented as, via $\widehat{XY}=\widehat
X\widehat Y$,
\begin{eqnarray}
\widehat{AZB}=\widehat A\widehat Z\widehat B=\Br{\begin{array}{cc}
      \re(A) & -\im(A) \\
      \im(A) & \re(A)
    \end{array}
}\Br{\begin{array}{cc}
      \re(Z) & -\im(Z) \\
      \im(Z) & \re(Z)
    \end{array}
}\Br{\begin{array}{cc}
      \re(B) & -\im(B) \\
      \im(B) & \re(B)
    \end{array}
}
\end{eqnarray}
Then from Eq.~\eqref{eq:AXB-delta}, we see that
\begin{eqnarray}
\delta(AZB)&=& \delta(\widehat{AZB})=\delta(\widehat A\widehat
Z\widehat B) =
\abs{\det^{-n}(\widehat A)\det^{-m}(\widehat B)}\delta(\widehat Z)\\
&=&\det^{-n}(AA^*)\det^{-m}(BB^*)\delta(Z).
\end{eqnarray}
The result is proven.
\end{proof}

\begin{prop}\label{prop:AXA*}
For $A\in\complex^{m\times m}$, and $m\times m$ Hermitian complex
matrix $X\in\complex^{m\times m}$, we have
\begin{eqnarray}
\delta(AXA^*) = \abs{\det(AA^*)}^{-m}\delta(X).
\end{eqnarray}
\end{prop}

\begin{proof}
By using the Fourier transform of the matrix delta function (see
Eq.~\eqref{eq:matrix-delta-function})
\begin{eqnarray}
\delta(AXA^*) = \frac1{2^m\pi^{m^2}}\int
e^{\mathrm{i}\Tr{TAXA^*}}[\dif T] = \frac1{2^m\pi^{m^2}}\int
e^{\mathrm{i}\Tr{A^*TAX}}[\dif T],
\end{eqnarray}
where $T=[t_{ij}]$ is also an $m\times m$ Hermitian complex matrix,
and $[\dif T]:=\prod_j\dif t_{jj}\prod_{i<j}\dif \re(t_{ij})\dif
\im(t_{ij})$.

Let $H=A^*TA$. Then we have (see Proposition 3.4 in
\cite{Zhang2015vol}):
\begin{eqnarray}
[\dif H] = \abs{\det(AA^*)}^m[\dif T].
\end{eqnarray}
Thus
\begin{eqnarray}
\delta(AXA^*) = \abs{\det(AA^*)}^{-m}\frac1{2^m\pi^{m^2}}\int
e^{\mathrm{i}\Tr{HX}}[\dif H]=\abs{\det(AA^*)}^{-m}\delta(X),
\end{eqnarray}
We are done.
\end{proof}

\section{Applications}

\subsection{Joint distribution of eigenvalues of Wishart matrix ensemble}

The first application is to calculate the joint distribution of
Wishart matrix ensemble \cite{James1964}. Note that the so-called
Wishart matrix ensemble is the set of all complex matrices of the
form $W=ZZ^\dagger$ with $Z$ an $m\times n(m\leqslant n)$ complex
Gaussian matrix, i.e., a matrix with all entries being standard
complex Gaussian random variables,
$$
\varphi(Z) = \frac1{\pi^{mn}}\exp\Pa{-\Tr{ZZ^\dagger}},
$$
the distribution density of $W$ is given by
\begin{eqnarray}\label{eq:Wishart-distribution}
\bP(W) &=& \int \delta\Pa{W-ZZ^\dagger} \varphi(Z)[\dif Z]\\
&=&\frac1{\pi^{mn}} \int
\delta\Pa{W-ZZ^\dagger}\exp\Pa{-\Tr{ZZ^\dagger}}[\dif Z].
\end{eqnarray}
That is,
\begin{eqnarray}
\bP(W) &=& \frac1{\pi^{mn}} e^{-\Tr{W}}\int
\delta\Pa{W-ZZ^\dagger}[\dif Z].
\end{eqnarray}
Let $Z=\sqrt{W}Y$. Then $ZZ^\dagger = \sqrt{W}YY^\dagger\sqrt{W}$
and $[\dif Z] = \det(W)^n[\dif Y]$. By Proposition~\ref{prop:AXA*},
we get
$$
\delta\Pa{W-ZZ^\dagger} =
\delta\Pa{\sqrt{W}\Pa{\I-YY^\dagger}\sqrt{W}} =
\frac1{\det(W)^m}\delta\Pa{\I-YY^\dagger}.
$$
Therefore
\begin{eqnarray}
\bP(W) &=& \frac1{\pi^{mn}} \det^{n-m}(W)e^{-\Tr{W}}\int
\delta\Pa{\I_m-YY^\dagger}[\dif Y]\propto \det^{n-m}(W)e^{-\Tr{W}},
\end{eqnarray}
where $\int \delta\Pa{\I-YY^\dagger}[\dif Y]:=C$ is a constant,
independent of $W$.

\begin{prop}\label{prop:stiefel}
It holds that
\begin{eqnarray}
\int \delta\Pa{\I_m-YY^\dagger}[\dif Y] =
\frac{\pi^{\frac12m(2n-m+1)}}{\prod^m_{k=1}(n-k)!}.
\end{eqnarray}
In particular, (i) for $m=1$,
\begin{eqnarray}
\int \delta\Pa{1-\inner{\bsy}{\bsy}}[\dif \bsy] =
\frac{\pi^n}{(n-1)!} = \frac{\pi^n}{\Gamma(n)}.
\end{eqnarray}
(ii) for $m=n$, we have
\begin{eqnarray}
\int \delta\Pa{\I_n-YY^\dagger}[\dif Y] =
\frac{\pi^{\frac12n(n+1)}}{\prod^n_{k=1}\Gamma(k)} =
2^{-n}\cdot\vol(\rU(n)).
\end{eqnarray}
\end{prop}

\begin{proof}
Since
\begin{eqnarray}
1=\int \bP(W)[\dif W] = \frac
C{\pi^{mn}}\int_{W>0}\det^{n-m}(W)e^{-\Tr{W}}[\dif W],
\end{eqnarray}
where via $W=UwU^\dagger$ for $w=\diag(w_1,\ldots,w_m)$ where
$w_1>\cdots>w_m$
$$
[\dif W] = \Delta(w)^2[\dif w][U^\dagger \dif U].
$$
It follows that
\begin{eqnarray}
1 = \frac C{\pi^{mn}}\vol(\rU(m)/\mathbb{T}^m)\int_{w_1>\cdots>w_m}
\Delta(w)^2\det^{n-m}(w)e^{-\Tr{w}}[\dif w].
\end{eqnarray}
By using Selberg integral formula, we see that
\begin{eqnarray}
\int \Delta(w)^2\det^{n-m}(w)e^{-\Tr{w}}[\dif w] =
\prod^m_{k=1}k!(n-k)!.
\end{eqnarray}
Thus
\begin{eqnarray}
\int_{w_1>\cdots>w_m} \Delta(w)^2\det^{n-m}(w)e^{-\Tr{w}}[\dif w]
=\frac1{m!} \prod^m_{k=1}k!(n-k)!.
\end{eqnarray}
We conclude that
\begin{eqnarray}
\int \delta\Pa{\I-YY^\dagger}[\dif Y] =
\frac{\vol(\mathbb{T}^m)m!\pi^{mn}}{\vol(\rU(m))\prod^m_{k=1}k!(n-k)!}.
\end{eqnarray}
Now
$$
\vol(\rU(m)) =
\frac{2^m\pi^{\frac{m(m+1)}2}}{\prod^m_{k=1}\Gamma(k)}\quad\text{and}\quad
\vol(\mathbb{T}^m) = (2\pi)^m.
$$
We obtain that
\begin{eqnarray}
\int \delta\Pa{\I-YY^\dagger}[\dif Y] =
\frac{\pi^{\frac12m(2n-m+1)}}{\prod^m_{k=1}(n-k)!}.
\end{eqnarray}
This completes the proof.
\end{proof}

\begin{cor}
It holds that
\begin{eqnarray}
\int \delta\Pa{W-ZZ^\dagger}[\dif Z] =
\frac{\pi^{\frac12m(2n-m+1)}}{\prod^m_{k=1}(n-k)!}\det(W)^{n-m}.
\end{eqnarray}
\end{cor}

\begin{proof}
It is easily seen that
\begin{eqnarray}
\int \delta\Pa{W-ZZ^\dagger}[\dif Z] &=&\det(W)^{n-m}\int
\delta\Pa{\I_m-YY^\dagger}[\dif Y]\\
&=& \frac{\pi^{\frac12m(2n-m+1)}}{\prod^m_{k=1}(n-k)!}\det(W)^{n-m}.
\end{eqnarray}
Here we used the result in Proposition~\ref{prop:stiefel}.
\end{proof}

\begin{remark}
Denote $\rU(m,n):=\Set{Z\in\complex^{m\times n}:
ZZ^\dagger=\I_m}(m\leqslant n)$. Note that
$$
\vol(\rU(m,n)) = \int \delta\Pa{\I_m - ZZ^\dagger}[Z^\dagger\dif Z]
= \frac{2^m\pi^{\frac12m(2n-m+1)}}{\prod^m_{k=1}(n-k)!}.
$$
This indicates that
$$
\int \delta\Pa{\I_m - ZZ^\dagger}[Z^\dagger\dif Z] = 2^m \int
\delta\Pa{\I_m - ZZ^\dagger}[\dif Z].
$$
In addition, the delta integral can be reformulated in terms of
another form:
$$
\int \delta\Pa{W - ZZ^\dagger}[\dif Z] = \int_{ZZ^\dagger=W}[\dif
Z].
$$
Finally, we have seen that
$$
\bP(W) =
\frac{\det(W)^{n-m}e^{-\Tr{W}}}{\pi^{\binom{m}{2}}\prod^m_{k=1}(n-k)!}.
$$
\end{remark}

\subsection{Joint distribution of eigenvalues of induced random quantum state ensemble}

Any mixed state $\rho$ (i.e., nonnegative complex matrix of
trace-one) acting on $\cH_m$, may be purified by finding a pure
state $\ket{X}$ in the composite Hilbert space $\cH_m\ot\cH_m$, such
that $\rho$ is given by the partial tracing over the auxiliary
subsystem,
\begin{eqnarray}
\ket{X}\longrightarrow \rho=\Ptr{2}{\proj{X}}.
\end{eqnarray}
In a loose sense, the purification corresponds to treating any
density matrix of size $m$ as a vector of size $m^2$.

Consider a bipartite $m\ot n$ composite quantum system. Pure states
of this system $\ket{X}$ may be represented by a rectangular complex
matrix $X$. The partial tracing with respect to the $n$-dimensional
subspace gives a reduced density matrix of size $m$:
$\rho=\Ptr{n}{\proj{X}}$. The natural measure in the space of
$mn$-dimensional pure states induces the measure $P_{m,n}(\rho)$ in
the space of the reduced density matrices of size $m$.

Without loss of generality, we assume that $m\leqslant n$, then
$\rho$ is \emph{generically} positive definite (here \red{something
is generic means it holds with probability one}). In any case, we
are only interested in the distribution of the positive eigenvalues.
Let us call the corresponding positive reduced density matrix again
$\rho = XX^\dagger$, where $X$ is a $m\times n$ matrix. First we
calculate the distribution of matrix elements
\begin{eqnarray}
\bP(\rho)\propto\int[\dif X] \delta\Pa{\rho -
XX^\dagger}\delta\Pa{1-\Tr{XX^\dagger}}
\end{eqnarray}
where the first delta function is a delta function of a Hermitian
matrix and in the second delta function $\Tr{XX^\dagger}$ may be
substituted by $\Tr{\rho}$. Since $\rho$ is positive definite we can
make a transformation
\begin{eqnarray}
X=\sqrt{\rho}\widetilde X,
\end{eqnarray}
it follows that $[\dif X]=\Pa{\det\rho}^n[\dif\widetilde X]$
\cite{Zhang2015vol}. The matrix delta function may now be written as
\begin{eqnarray}
\delta\Pa{\sqrt{\rho}\Pa{\I-\widetilde X\widetilde
X^\dagger}\sqrt{\rho}} = (\det\rho)^{-m}\delta\Pa{\I-\widetilde
X\widetilde X^\dagger}.
\end{eqnarray}
As the result the
distribution of matrix elements is given by
\begin{eqnarray}
\bP(\rho)\propto\theta(\rho)
\delta\Pa{1-\Tr{\rho}}\Pa{\det\rho}^{n-m}
\end{eqnarray}
where the theta function assures that $\rho$ is positive definite.
It is then easy to show by the methods of random matrix theory that
the joint density of eigenvalues
$\Lambda=\set{\lambda_1,\ldots,\lambda_m}$ of $\rho$ is given by
\begin{eqnarray}
\bP_{m,n}(\lambda_1,\ldots,\lambda_m)\propto\delta\Pa{1-\sum^m_{j=1}\lambda_j}\prod^m_{j=1}\lambda^{n-m}_j\theta(\lambda_j)
\prod_{1\leqslant i<j\leqslant m}(\lambda_i-\lambda_j)^2.
\end{eqnarray}
This result is firstly obtained by \.{Z}yczkowski \cite{Karol01} in
2001.

\subsection{Joint distribution of eigenvalues of two Wishart matrices}

The second application is to calculate the distribution of the sum
of a finite number of complex Wishart matrices taken from the same
Wishart ensemble \cite{James1964}. The distribution of the sum of
two Wishart matrices is considered in \cite{Kumar2014}. Let us
consider two independent complex matrices $A$ and $B$ of dimensions
$m\times n_A$ and $m\times n_B$ taken, respectively, from the
distributions
\begin{eqnarray}
\mathbf{P}_A(A) &=&
\frac1{\pi^{mn_A}\det^{n_A}(\Sigma_A)}\exp\Pa{-\Tr{\Sigma^{-1}_AAA^\dagger}},\\
\mathbf{P}_B(B) &=&
\frac1{\pi^{mn_B}\det^{n_B}(\Sigma_B)}\exp\Pa{-\Tr{\Sigma^{-1}_BBB^\dagger}}.
\end{eqnarray}
Here $\Sigma_A,\Sigma_B$ are the covariance matrices. Since the
domains of $A$ and $B$ remain invariant under unitary rotation,
without loss of generality, we may take
$\Sigma_A=\diag(\sigma_{A1},\ldots,\sigma_{Am})$ and
$\Sigma_B=\diag(\sigma_{B1},\ldots,\sigma_{Bm})$. We assume that
$m\leqslant n_A,n_B$. And we have
$$
\int[\dif A]\mathbf{P}_A(A)=1\quad\text{and}\quad\int[\dif
B]\mathbf{P}_B(B)=1.
$$
The matrix $AA^\dagger$ and $BB^\dagger$ are then $n$-variate
complex-Wishart-distributed, i.e., $AA^\dagger\sim
W^\complex_m(n_A,\Sigma_A)$ and $BB^\dagger\sim
W^\complex_m(n_B,\Sigma_B)$.

Ones are interested in the statistics of the ensemble of $m\times m$
Hermitian matrices
$$
W = AA^\dagger+BB^\dagger.
$$
The distribution of $W$ can be obtained as
\begin{eqnarray}
\mathbf{P}_W(W) = \int[\dif A]\int[\dif B] \delta\Pa{W - AA^\dagger
- BB^\dagger}\mathbf{P}_A(A)\mathbf{P}_B(B).
\end{eqnarray}
In what follows, our method will be different from Kumar's in
\cite{Kumar2014}. Let $A=\sqrt{W}\widetilde A$ and
$B=\sqrt{W}\widetilde B$. Then $[\dif A] = \det(W)^{n_A}[\dif
\widetilde A]$ and $[\dif B] = \det(W)^{n_B}[\dif \widetilde B]$. We
also have
$$
\delta\Pa{W - AA^\dagger - BB^\dagger} = \det(W)^{-m}\delta\Pa{\I -
\widetilde A\widetilde A^\dagger - \widetilde B\widetilde
B^\dagger}.
$$
Then
\begin{eqnarray*}
\mathbf{P}_W(W) &=&
\frac{\det^{n_A+n_B-m}(W)}{\pi^{m(n_A+n_B)}\det^{n_A}(\Sigma_A)\det^{n_B}(\Sigma_B)}\int[\dif
\widetilde A]\int[\dif \widetilde B] \\
&&\times\delta\Pa{\I - \widetilde A\widetilde A^\dagger - \widetilde
B\widetilde
B^\dagger}e^{-\Tr{\sqrt{W}\Sigma^{-1}_A\sqrt{W}\widetilde
A\widetilde A^\dagger}-\Tr{\sqrt{W}\Sigma^{-1}_B\sqrt{W}\widetilde
B\widetilde B^\dagger}}.
\end{eqnarray*}
That is,
\begin{eqnarray}
\mathbf{P}_W(W) &\propto&\det^{n_A+n_B-m}(W) \notag\\
&&\times \int[\dif \widetilde A]\int[\dif \widetilde B] \delta\Pa{\I
- \widetilde A\widetilde A^\dagger - \widetilde B\widetilde
B^\dagger}e^{-\Tr{\sqrt{W}\Sigma^{-1}_A\sqrt{W}\widetilde
A\widetilde A^\dagger+\sqrt{W}\Sigma^{-1}_B\sqrt{W}\widetilde
B\widetilde B^\dagger}}.
\end{eqnarray}
Now consider the reduction of the sum
$\Tr{\sqrt{W}\Sigma^{-1}_A\sqrt{W}\widetilde A\widetilde
A^\dagger+\sqrt{W}\Sigma^{-1}_B\sqrt{W}\widetilde B\widetilde
B^\dagger}$. Since $\I = \widetilde A\widetilde A^\dagger+\widetilde
B\widetilde B^\dagger$, it follows that
\begin{eqnarray}
&&\Tr{\sqrt{W}\Sigma^{-1}_A\sqrt{W}\widetilde A\widetilde
A^\dagger}+\Tr{\sqrt{W}\Sigma^{-1}_B\sqrt{W}\widetilde B\widetilde
B^\dagger}\\
&&= \Tr{W\Sigma^{-1}_A} +
\Tr{\sqrt{W}(\Sigma^{-1}_B-\Sigma^{-1}_A)\sqrt{W}\widetilde
B\widetilde B^\dagger}\\
&&=\Tr{W\Sigma^{-1}_B} +
\Tr{\sqrt{W}(\Sigma^{-1}_A-\Sigma^{-1}_B)\sqrt{W}\widetilde
A\widetilde A^\dagger}.
\end{eqnarray}
Hence
\begin{eqnarray}
\mathbf{P}_W(W)
&\propto&\frac{\det^{n_A+n_B-m}(W)}{e^{\Tr{W\Sigma^{-1}_A}}}Q(W),
\end{eqnarray}
where
\begin{eqnarray*}
Q(W)&:=&\int[\dif \widetilde A]\int[\dif \widetilde B] \delta\Pa{\I
- \widetilde A\widetilde A^\dagger - \widetilde B\widetilde
B^\dagger}e^{\Tr{\sqrt{W}(\Sigma^{-1}_A-\Sigma^{-1}_B)\sqrt{W}\widetilde
B\widetilde B^\dagger}}\\
&=&\int_{0<\widetilde A\widetilde A^\dagger<\I_m}[\dif \widetilde
A]e^{\Tr{\sqrt{W}(\Sigma^{-1}_A-\Sigma^{-1}_B)\sqrt{W}(\I -
\widetilde A\widetilde A^\dagger)}}\\
&=&e^{\Tr{W(\Sigma^{-1}_A-\Sigma^{-1}_B)}}\int_{0<\widetilde
A\widetilde A^\dagger<\I_m}[\dif \widetilde
A]e^{\Tr{\sqrt{W}(\Sigma^{-1}_B-\Sigma^{-1}_A)\sqrt{W}\widetilde
A\widetilde A^\dagger}}
\end{eqnarray*}
That is,
\begin{eqnarray}
\mathbf{P}_W(W)
&\propto&\frac{\det^{n_A+n_B-m}(W)}{e^{\Tr{W\Sigma^{-1}_B}}}\int_{0<\widetilde
A\widetilde A^\dagger<\I_m}[\dif \widetilde
A]e^{\Tr{\sqrt{W}(\Sigma^{-1}_B-\Sigma^{-1}_A)\sqrt{W}\widetilde
A\widetilde A^\dagger}}.
\end{eqnarray}
Similarly, we have
\begin{eqnarray}
\mathbf{P}_W(W)
&\propto&\frac{\det^{n_A+n_B-m}(W)}{e^{\Tr{W\Sigma^{-1}_A}}}\int_{0<\widetilde
B\widetilde B^\dagger<\I_m}[\dif \widetilde
B]e^{\Tr{\sqrt{W}(\Sigma^{-1}_A-\Sigma^{-1}_B)\sqrt{W}\widetilde
B\widetilde B^\dagger}}.
\end{eqnarray}
In fact, we have
\begin{eqnarray*}
&&\frac1{e^{\Tr{W\Sigma^{-1}_B}}}\int_{0<\widetilde A\widetilde
A^\dagger<\I_m}[\dif \widetilde
A]e^{\Tr{\sqrt{W}(\Sigma^{-1}_B-\Sigma^{-1}_A)\sqrt{W}\widetilde
A\widetilde
A^\dagger}}\\
&&=\frac1{e^{\Tr{W\Sigma^{-1}_A}}}\int_{0<\widetilde B\widetilde
B^\dagger<\I_m}[\dif \widetilde
B]e^{\Tr{\sqrt{W}(\Sigma^{-1}_A-\Sigma^{-1}_B)\sqrt{W}\widetilde
B\widetilde B^\dagger}}.
\end{eqnarray*}
Let $\Lambda =\sqrt{W}(\Sigma^{-1}_A-\Sigma^{-1}_B)\sqrt{W}$. By
Gram-Schmidt orthogonalization procedure to write $\widetilde A=
TU_1$, where $U_1$ is $m\times n_A$ matrix such that
$U_1U^\dagger_1=\I_m$ and $T$ is a $m\times m$ lower triangular
matrix with diagonal entries real and positive. Then
$$
[\dif \widetilde A] = \prod^m_{j=1}t^{2(n_A-j)+1}_{jj}[\dif T][\dif
U_1 U^\dagger],
$$
where $U$ is the enlarged $n_A\times n_A$ unitary matrix of $U_1$.
Let $X=\widetilde A\widetilde A^\dagger = TT^\dagger$. We then have
\cite{Zhang2015vol}:
$$
[\dif X] = 2^m\prod^m_{j=1}t^{2(m-j)+1}_{jj}[\dif T].
$$
Combining together the above results gives that \cite{Zhang2015vol}
\begin{eqnarray}
[\dif \widetilde A] = 2^{-m}\det^{n_A-m}(X)[\dif X][\dif U_1
U^\dagger].
\end{eqnarray}
This means that
\begin{eqnarray}
\int_{0<\widetilde A\widetilde A^\dagger<\I_m}[\dif \widetilde
A]e^{\Tr{\Lambda \widetilde A \widetilde A^\dagger}} &=&
2^{-m}\int^{\I_m}_0 [\dif X]\det^{n_A-m}(X)e^{\Tr{\Lambda
X}}\int[\dif
U_1 U^\dagger]\\
&\propto&\int^{\I_m}_0 [\dif X]\det^{n_A-m}(X)e^{\Tr{\Lambda X}}.
\end{eqnarray}
Similarly,
\begin{eqnarray}
\int_{0<\widetilde B\widetilde B^\dagger<\I_m}[\dif \widetilde
B]e^{-\Tr{\Lambda \widetilde B \widetilde
B^\dagger}}\propto\int^{\I_m}_0 [\dif
X]\det^{n_B-m}(X)e^{-\Tr{\Lambda X}}.
\end{eqnarray}

\begin{thrm}
It holds that
\begin{eqnarray}
\bP_W(W) \propto
\frac{\det^{n_A+n_B-m}(W)}{e^{\Tr{W\Sigma^{-1}_B}}}\int^{\I_m}_0
[\dif X]\det^{n_A-m}(X)e^{\Tr{\Lambda X}}
\end{eqnarray}
and
\begin{eqnarray}
\bP_W(W) \propto
\frac{\det^{n_A+n_B-m}(W)}{e^{\Tr{W\Sigma^{-1}_A}}}\int^{\I_m}_0
[\dif X]\det^{n_B-m}(X)e^{-\Tr{\Lambda X}}.
\end{eqnarray}
Moreover
\begin{eqnarray}
\frac1{e^{\Tr{W\Sigma^{-1}_B}}}\int^{\I_m}_0 [\dif
X]\det^{n_A-m}(X)e^{\Tr{\Lambda
X}}=\frac1{e^{\Tr{W\Sigma^{-1}_A}}}\int^{\I_m}_0 [\dif
X]\det^{n_B-m}(X)e^{-\Tr{\Lambda X}}.
\end{eqnarray}
In particular, if $\Lambda=0$ (i.e., $\Sigma_A=\Sigma_B:=\Sigma$),
then
\begin{eqnarray}
\bP_W(W) \propto \frac{\det^{n_A+n_B-m}(W)}{e^{\Tr{W\Sigma^{-1}}}}.
\end{eqnarray}
\end{thrm}
In the following, we calculate, via $X=U\diag (x)U^\dagger$ where
$\diag(x) = \diag(x_1,\ldots,x_m)$,
\begin{eqnarray}
\int^{\I_m}_0 [\dif X]\det^{n_A-m}(X)e^{\Tr{\Lambda X}} \propto
\int^1_0\cdots\int^1_0 \Delta(x)^2\prod^m_{j=1}x^{n_A-m}_j\dif x_j
\int\dif\mu_{\mathrm{Haar}}(U)e^{\Tr{\Lambda U\diag (x)U^\dagger}}
\end{eqnarray}
W.l.o.g, we assume that $\Lambda$ is in the diagonal form, i.e.
$\Lambda = \diag(\lambda_1,\ldots,\lambda_m)$ where each
$\lambda_j\in\real$. By Harish-Chandra-Itzykson-Zuber integral
formula \cite{Harish1958,Itzykson1980},
\begin{eqnarray}
\int_{\rU(m)}\dif\mu_{\mathrm{Haar}}(U)e^{\Tr{A UBU^\dagger}} =
\Pa{\prod^m_{j=1}\Gamma(j)}\frac{\det\Br{\exp(a_ib_j)}}{\Delta(a)\Delta(b)},
\end{eqnarray}
it follows that
\begin{eqnarray}
\int\dif\mu_{\mathrm{Haar}}(U)e^{\Tr{\Lambda U\diag (x)U^\dagger}} =
\Pa{\prod^m_{j=1}\Gamma(j)}\frac{\det\Br{\exp(\lambda_ix_j)}}{\Delta(\lambda)\Delta(x)}.
\end{eqnarray}
Therefore
\begin{eqnarray}
\int^{\I_m}_0 [\dif X]\det^{n_A-m}(X)e^{\Tr{\Lambda X}} &\propto&
\frac1{\Delta(\lambda)}\int^1_0\cdots\int^1_0
\Delta(x)\det\Br{\exp(\lambda_ix_j)}\prod^m_{j=1}x^{n_A-m}_j\dif
x_j\\
&\propto&\frac1{\Delta(\lambda)}\int^1_0\cdots\int^1_0
\abs{\Delta(x)}\prod^m_{j=1}e^{\lambda_jx_j}x^{n_A-m}_j\dif x_j,
\end{eqnarray}
implying that
\begin{eqnarray}
\int^{\I_m}_0 [\dif X]\det^{n_A-m}(X)e^{\Tr{\Lambda X}}
&\propto&\frac1{\Delta(\lambda)}\Pa{\prod^m_{j=1}\partial^{n_A-m}_{\lambda_j}}\int^1_0\cdots\int^1_0
\abs{\Delta(x)}\prod^m_{j=1}e^{\lambda_jx_j}\dif x_j.
\end{eqnarray}

\begin{remark}
Kumar in \cite{Kumar2014} have presented analytical formula for
distribution of the sum of two Wishart matrices in terms of the
confluent hypergeometric function of matrix argument
\cite{Mathai1997}, which is defined by
\begin{eqnarray}
_1F_1(a;c;-\Lambda) =
\frac{\Gamma_p(c)}{\Gamma_p(a)\Gamma_p(c-a)}\int^{\I_p}_0
\det(X)^{a-p}\det(\I-X)^{c-a-p}e^{-\Tr{\Lambda X}}[\dif X],
\end{eqnarray}
where
$\Gamma_p(a):=\pi^{\frac{p(p-1)}2}\Gamma(a)\Gamma(a-1)\cdots\Gamma(a-p+1)$.
Kummer's formula for the confluent hypergeometric function $_1F_1$
is given by
\begin{eqnarray}
_1F_1(a;c;-\Lambda) = e^{-\Tr{\Lambda}}\cdot\ _1F_1(c-a;c;\Lambda).
\end{eqnarray}

\end{remark}

\begin{remark}
We can present a simple approach to the similar result. Denote
$W^\complex_m(n,\Sigma)$ Wishart matrix ensemble for which each
matrix is of the form $W=AA^*$, where $A\in\complex^{m\times
n}(m\leqslant n)$. Now $W=W_1+W_2$, where $W_1,W_2\in
W^\complex_m(n,\Sigma)$, can be rewritten as
$$
W=ZZ^\dagger, Z=[A,B]\in\complex^{m\times 2n}.
$$
Then $W\sim W^\complex_m(2n,\Sigma)$. Thus
$$
\mathbf{P}_W(W)\propto \det^{2n-m}(W)e^{-\Tr{\Sigma^{-1} W}}.
$$
The distribution of sum of an arbitrary finite number of Wishart
matrices can be derived as:
$$
\mathbf{P}_W(W)\propto \det^{kn-m}(W)e^{-\Tr{\Sigma^{-1} W}},
$$
where $W=\sum^k_{j=1}W_j$ for $W_j\sim W^\complex_m(n,\Sigma)$.
\end{remark}




\begin{thebibliography}{99}


\bibitem{Attila2017}
A. Lovas, A. Andai, {\em Volume of the space of qubit-qubit channels
and state transformations under random quantum channels},
\href{https://arxiv.org/abs/1708.07387}{arXiv: 1708.07387}

\bibitem{Christandl2014}
M. Christandl, B. Doran, S. Kousidis, M. Walter,
\newblock {\em
Eigenvalue distributions of reduced density matrices},
\newblock \cmp
\href{http://dx.doi.org/10.1007/s00220-014-2144-4}{\textbf{332},
1-52 (2014).}

\bibitem{Harish1958}
Harish-Chandra, \newblock {\em Spherical Functions on a Semisimple
Lie Group, I}, \newblock Amer. J. Math.
\href{http://dx.doi.org/10.2307/2372786}{\textbf{80}, 241-310
(1958).}


\bibitem{Hoskins2009}
R.F. Hoskins,
\newblock {\em Delta Functions: An Introduction to
Generalised Functions},
\newblock 2nd Edition 2009, Woodhead Publishing Limited, Oxford (2011).

\bibitem{Itzykson1980}
C. Itzykson and J.B. Z\"{u}ber, \newblock{\em The planar
approximation. II}, \newblock \jmp
\href{http://dx.doi.org/10.1063/1.524438}{\textbf{21}, 411-421
(1980).}

\bibitem{James1964}
A.T. James, \newblock{\em Distributions of Matrix Variates and
Latent Roots Derived from Normal Samples}, \newblock Ann. Math.
Stat. \href{http://www.jstor.org/stable/2238504}{\textbf{35}(2),
475-501 (1964).}

\bibitem{Kumar2014}
S. Kumar,
\newblock {\em Eigenvalue statistics for the sum of two complex Wishart matrices},
\newblock Europhys. Lett.
\href{http://dx.doi.org/10.1209/0295-5075/107/60002}{\textbf{107},
60002 (2014).}

\bibitem{Mathai1997}
A.M. Mathai,
\newblock {\em Jacobians of Matrix Transformations and Functions of Matrix Arguments},
\newblock World Scientific (1997).

\bibitem{Mejia2017}
J. Mej\'{i}a, C. Zapata, A. Botero, \newblock{\em The difference
between two random mixed quantum states: exact and asymptotic
spectral analysis}, \newblock \jpa: Math. Theor.
\href{http://dx.doi.org/10.1088/1751-8121/50/2/025301}{\textbf{50},
025301 (2017).}

\bibitem{Moghadasi2012}
S.R. Moghadasi, {\em Polar decomposition of the $k$-fold product of
Lebesgue measure on $\real^n$}, Bull. Aust. Math. Soc.
\href{http://dx.doi.org/10.1017/S0004972711003273}{\textbf{85},
315-324 (2012).}

\bibitem{Zhang2015vol}
L. Zhang, {\em Volumes of orthogonal groups and unitary groups},
\href{http://arxiv.org/abs/1509.00537}{arXiv:1509.00537}

\bibitem{Zhang2016}
L. Zhang, {\em Average coherence and its typicality for random mixed
quantum states}, \jpa: Math. Theor.
\href{http://dx.doi.org/10.1088/1751-8121/aa6179}{\textbf{50},
155303 (2017).}

\bibitem{Karol01}
K. \.{Z}yczkowski, W. S{\l}omczynski,
\newblock {\em Induced measures in the space of mixed quantum states},
\newblock \jpa: Math. Gen.
\href{http://dx.doi.org/10.1088/0305-4470/34/35/335}{\textbf{34},
7111 (2001).}



\end{thebibliography}
\end{document}